\documentclass[a4paper,11pt]{article}

\usepackage{lineno}
\usepackage{xcolor}
\usepackage{booktabs}
\usepackage{longtable}
\usepackage{algpseudocode}
\usepackage{appendix}
\usepackage[normalem]{ulem}

\usepackage{jheppub} 

\usepackage{amsmath}
\usepackage{amsthm}
\usepackage{mathtools}
\usepackage{algorithm}
\usepackage{algpseudocode}
\newtheorem{theorem}{Theorem}
\newtheorem{definition}{Definition}
\newtheorem{lemma}{Lemma}
\newtheorem{drule}{Rule}

\definecolor{SHCcolor}{rgb}{0.9,0.1,0.1}

\DeclareMathOperator{\area}{area}

\title{Properties of the contraction map for holographic entanglement entropy inequalities}

\author[a,b]{Ning Bao,}
\author[a,c]{and Joydeep Naskar}
\affiliation[a]{Department of Physics, Northeastern University, Boston, MA 02115, U.S.A.}
\affiliation[b]{Computational Science Initiative, Brookhaven
National Laboratory, Upton, NY 11973 U.S.A.}
\affiliation[c]{The NSF AI Institute for Artificial Intelligence and Fundamental Interactions, Cambridge, MA, U.S.A.}

\emailAdd{ningbao75@gmail.com}
\emailAdd{naskar.j@northeastern.edu}

\abstract{We present a deterministic way of finding contraction maps for candidate holographic entanglement entropy inequalities modulo choices due to actual degeneracy. We characterize its complexity and give an argument for the completeness of the contraction map proof method as a necessary and sufficient condition for the validity of an entropy inequality for holographic entanglement.}

\makeatletter\def\@fpheader{~}\makeatother

\begin{document}
\maketitle
\flushbottom

\section{Introduction}
Holographic entropy inequalities constrain the allowed quantum states in conformal field theory that are consistent with the existence of a semi-classical gravity dual via the Ryu-Takayanagi formula \cite{Ryu_2006}. These constraints on entanglement are stronger than those obeyed by generic quantum states, as first discovered in \cite{Hayden_2013}. The program of characterizing them systematically was initiated in \cite{Bao_2015} and continued in \cite{Hern_ndez_Cuenca_2019, Avis_2023}, where in particular the classification of the set of entropic constraints was completed for mixed states on five parties. Further work towards characterizing holographic entanglement includes the study of bit threads \cite{Freedman_2016, Headrick_2018, Hubeny:2018bri, Cui_2019, Harper_2019, headrick2022covariant}, the connection to quantum marginal independence problems \cite{Hernandez-Cuenca:2019jpv,Hernandez-Cuenca:2022pst,He:2022bmi,He:2023cco}, generalizations to other measures of holographic entanglement \cite{Umemoto_2018, Nguyen_2018, bao2018holographic, bao2019conditional, dutta2019canonical}, and extensions of the tools and proof techniques beyond the holographic setting \cite{Bao_2020, bao2021gap, Walter_2021, Bao_2022}.

An essential tool introduced by \cite{Bao_2015} is a combinatorial strategy for proving the validity of holographic entropy inequalities, dubbed the contraction map proof method.\footnote{An alternative method based on the bit-thread formulation of \cite{Freedman_2016} may offer a conceptually clearer interpretation of the underlying meaning of the inequalities. Unfortunately, as currently understood, its scope of applicability is very limited \cite{Headrick:2020gyq}.}\footnote{Also, by appealing to a related object, the holographic cone of averaged entropies \cite{czech2022holographic,Fadel:2021urx}, a seven-party inequality was found in \cite{Czech_2023}.} The method states that an inequality is valid if there exists a certain map obeying a contraction property dictated by the inequality (this will be discussed in more detail in Section \ref{sec:method}). Though powerfully general, its inner workings are poorly understood, resulting in highly inefficient implementations of it which prove to be onerous, when not practically impossible to run, for inequalities involving more than five parties.
Up until now, the method employed to discover these contraction maps has been based on greedy algorithms, heavily reliant on the aid of heuristics, that brute-force search for global solutions to the contraction property.\footnote{Such an approach is prone to require backtracking, and indeed is the naive approach to solve the worst-case NP-complete constraint satisfaction problem, of which this problem is an instance. The central thesis of this work is that this particular instance of the constraint satisfaction problem is, in fact, not a worst-case instance, and thus admits computational speedup.} This is highly unsatisfactory not just at the practical level, but also at the fundamental one: the contraction condition is intrinsically geometric, and a better understanding of it should not only improve the proof method but also our holographic interpretation of these inequalities.

In this work, we develop a new deterministic technique for finding contraction maps for candidate inequalities. In particular, we find and prove a set of rules that provide unique solutions to the contraction condition that determines the map, thereby avoiding the arbitrary local choices that greedy approaches suffer from. We empirically demonstrate that our set of rules is strong enough that no backtracking is needed on any of the heretofore-proven inequalities for five, six and seven parties. We characterize the algorithmic complexity of implementing these deterministic rules and provide an argument that the contraction map proof method is a complete proof technique. In so doing, we show that this deterministic implementation of the contraction map proof technique provides a notable improvement both performance-wise and conceptually to the existing greedy approaches. In addition to the exponential speed-up, our algorithm can easily handle inequalities that are computationally intractable by the state-of-the-art greedy algorithm.

The organization of this paper is as follows. In Section \ref{sec:method}, we review relevant aspects of the contraction map proof method as applied to holographic entanglement entropy. In Section \ref{sec:rules}, we introduce the new technical notions needed for the new deterministic techniques, state those techniques, and prove that they must be respected by all contraction maps. In Section \ref{sec:complexity}, we characterize the algorithmic complexity of implementing these techniques and give an argument for the completeness of the contraction map proof method in general. Finally, in Section \ref{sec:discussion} we conclude with a discussion of our results and future directions. 

\section{Contraction Map Proofs for Holographic Entanglement Entropy}
\label{sec:method}
We first briefly introduce the basic ingredients and notation necessary to state the contraction map proof method for holographic entanglement entropy inequalities (see \cite{Bao_2015, Akers:2021lms, Avis_2023} for more details and intuition).
In holography, the entropy $S(X)$ of a boundary region $X$ is given by 
\begin{equation}
    S(X)=\frac{\area \mathcal{X}}{4G_N},
\end{equation}
where $\mathcal{X}$ is the Ryu-Takayanagi (RT) surface for $X$, and $G_N$ is Newton's constant \cite{Ryu_2006}.\footnote{Here we work only with the geometric contribution that gives the von Neumann entropy of boundary regions to leading order in $G_N$, and assume our boundary regions lie at moments of time symmetry such that RT applies.} As usual, the RT surface obeys a homology condition by which there exists a homology surface $W_X$ whose boundary $\partial W_X = X\cup\mathcal{X}$.
Given a set of $n\ge1$ boundary subsystems $\{X_i\}_{i=1}^n$, one can consider $2^n-1$ distinct subsystems $X_I\equiv \bigcup_{i\in I} X_i$, one for each non-empty subset $I\subseteq[n]\equiv\{1,\dots,n\}$. Using the shorthand $S(I)\equiv S(X_I)$, entropy inequalities can then be canonically written in the form
\begin{equation}
\label{eq:ineq}
    \sum_{l=1}^M \alpha_l S({I_l}) \geq \sum_{r=1}^N \beta_r S({J_r}),
\end{equation}
where $\alpha_l, \beta_r\ge1$ by convention, and the $I_l, J_r$ subsets are all distinct.

One can map the geometric picture of the Cauchy slice of bulk geometry to a graph picture and vice-versa where each distinct subsystem in the partitioned geometry is assigned to a vertex on the graph. An entanglement entropy inequality that holds on the graph also holds for the holographic geometry. From here on, we will refer to such inequalities as \emph{entropy inequalities on graphs}. The problem of proving an entropy inequality candidate for a holographic geometry can thus be translated into the problem of finding a contraction map from a hypercube representing the left-hand-side(LHS) to a hypercube representing the right-hand-side(RHS) of an inequality.

For the contraction map proof method, we consider the RT surfaces on the LHS, $\{\mathcal{X}_{I_l}\}_{l=1}^N$ and use them to partition the bulk Cauchy slice where they lie into bulk regions. The resulting bulk regions $W_x$ can be uniquely labelled by bitstrings $x\in\{0,1\}^N$ via the following inclusion/exclusion scheme:
\begin{equation}
    x_l = 
    \begin{cases}
        1, \qquad \text{if the region lies on the homology surface of $X_{I_l}$},\\
        0, \qquad \text{otherwise.}
    \end{cases}
\end{equation}
A particularly relevant set of bitstrings are those labeling regions adjacent to boundary subsystems; for each $i\in[n+1]$,
\begin{equation}
    x^{(i)}_l = \begin{cases}
        1, \qquad \text{if} \quad i\in I_l \\
        0, \qquad \text{otherwise,}
    \end{cases}
\end{equation}
where the ($n+1$)-th bitstring is all zeroes and is associated with the purifier. These bitstrings are often referred to as occurrence bitstrings and can be analogously defined for RHS subsystems using bitstrings $y\in\{0,1\}^M$.
Bitstrings differing by a single bit label adjacent bulk regions sharing a portion of RT surface. This way, this labeling not only encodes all regions that make up homology surfaces but also all portions that make up their RT surfaces. The idea of the contraction map proof method is to build a map that takes the LHS bulk regions $W_x$ and uses them to construct homology regions for RHS subsystems. The purpose of the contraction property is to make sure that no portion of RT surface of any LHS subsystem appears on the boundary of the newly constructed RHS homology regions more times than on the LHS ones. This guarantees that the areas bounding the resulting RHS homology regions are no larger than those of LHS RT surfaces. These homology regions for RHS subsystems will be bounded by bulk surfaces of not-necessarily-minimal area, so it will also be guaranteed that actual RT surfaces for RHS subsystems will have no larger area. The upshot is that if it is possible to find such a contraction map, then the inequality is valid for any possible holographic configuration. Explicitly, defining a weighted Hamming distance on length-$M$ bitstrings via
\begin{equation}
\label{eq:weightedham}
    d_\gamma(x,x') \equiv \sum_{i=1}^M \gamma_i \, |{x_i - x'_i}|,
\end{equation}
the contraction map proof method can be stated as follows:
\begin{theorem}
\label{thm:contraction}
    If there exists a map $f:\{0,1\}^M\to\{0,1\}^N$ with the homology property
    \begin{equation}
    \label{eq:homocond}
        f(x^{(i)})_r = 1 \qquad \text{iff} \quad i\in J_r, \qquad \forall \, i \in[n+1], \, r\in[N],
    \end{equation}   
    such that the following contraction conditions are obeyed,
    \begin{equation}
    \label{eq:contcond}
        d_\alpha(x,x') \geq d_\beta(f(x),f(x')) \qquad \forall \, x, x'\in\{0,1\}^M,
    \end{equation}
    then \eqref{eq:ineq} is a valid entropy inequality on graphs.
\end{theorem}

As explained above, $f$ builds homology regions for RHS subsystems using regions of those for LHS ones; in particular, the homology region for $J_r$ is given by $\bigcup_{f(x)_r = 1} W_x$, i.e., the union of all LHS regions labeled by an $x$ such that $f(x)_r = 1$. The homology property \eqref{eq:homocond} on occurrence bitstrings thus simply makes sure that the constructed homology regions are adjacent to the desired RHS subsystem. This homology property can essentially be understood as providing the initial conditions for the problem of finding a contraction map. The challenge of obeying all the contraction conditions set by \eqref{eq:contcond} is where the complexity of the proof method lies.

To study the contraction maps relevant to Theorem \ref{thm:contraction}, it is useful to think of their domain and co-domain spaces as unit hypercubes $H_M\equiv\{0,1\}^M$ and $H_N\equiv\{0,1\}^N$ respectively, where every bitstring labels a corresponding vertex. In such hypercubes, Hamming distances between any pair of vertices are given by the minimal distance between them following edges of the hypercube (cf. using the taxicab metric), which is also equivalent to just counting the number of different bits between the bitstrings labeling the vertices. The weighted Hamming distance introduced in \eqref{eq:weightedham} that is relevant for the contraction condition in \eqref{eq:contcond} simply scales each dimension of the hypercube by a multiplicative factor\footnote{It should be noted that this weighted Hamming distance should only be used for the left-hand side hypercube of a candidate inequality. For the right-hand side, terms with an integer coefficient greater than one must be repeated that number of times for the contraction proof method to work.}. Because bitstrings uniquely label hypercube vertices, we will refer to the two interchangeably in what follows.~\footnote{For holographers, the sub-graph of the RHS hypercube that is the image of the LHS hypercube (modulo the unphysical vertices noted above) forms a desiccation of a multi-boundary wormhole geometry, where the boundary regions are those labeled by the single characters. This graph is built such that each vertex is a bulk region, and each edge is a wormhole connecting adjacent bulk regions corresponding to vertices with Hamming distance one. This connection will not be essential to understanding the present work but is given here to make the work more self-contained.}

The above constraints are the only ones that must be satisfied for a contraction map to exist for a given holographic entanglement entropy inequality.
The traditional way of finding such contraction maps is via a greedy algorithm, wherein the map is built recursively through locally optimal choices of image bitstrings. In particular, one initially fixes some ordering of the domain bitstrings, with occurrence bitstrings first as they provide a set of initial conditions for the contraction problem.\footnote{Subsequent input bitstrings may be ordered canonically e.g. as binary numbers, or by minimizing Hamming distances with the occurrence bitstrings to produce the tightest contraction conditions first. The latter method turns out to be crucial in making the greedy algorithm succeed.} One then attempts to find an image bitstring for the next input bitstring that obeys all contraction conditions with the previous ones. If more than one solution exists, one of them is picked at random and the rest are stored. This step is iterated with subsequent input bitstrings until either a full contraction map is found, or one hits an input for which no output satisfies the contraction conditions. The latter generically happens (even if a contraction map does exist) because solutions are only locally obtained following some ordering which accounts for previous constraints, but not future ones and random choices are made whenever more than one option is available. As a result, such a failure is only local and simply requires revisiting previous solution choices and re-iterating the process for every such choice made. A definitive failure to find a contraction map, with which one can conclude no such map exists, occurs when all possible solution choices for all input bitstrings have been exhausted with local failures in all cases. Otherwise, a contraction map will always be found.

This potential for assigning bitstrings incorrectly is the main downside of the greedy algorithm; it is assigning bits at times where it is not clear if those bits are free to be assigned to be 1 or 0. As such, it oftentimes requires backtracking to previous solutions, where incorrectly assigned bits need to be flipped. Finding which bits are incorrectly assigned is an algorithmically time-consuming process, which only gets worse as the inequality gets more terms. In fact, as the greedy algorithm does not exploit known structures of holographic entanglement entropies, its algorithmic complexity should be similar to that of $3-SAT$, as it is approaching the problem as if it were an unstructured constraint satisfaction problem. As such, it would be a worst-case NP-complete algorithm. This would clearly be too slow to scale beyond a relatively small $N$ and $M$.

\section{Deterministic Approach to Contraction Maps}
\label{sec:rules}

We will bypass the need for a greedy algorithm by directly solving deterministically for image bits that are uniquely fixed by the initial homology property and the requirements of the contraction conditions. The outcome of doing so can be one of three:
\begin{enumerate}
    \item All image bits are fixed, thereby yielding a complete contraction map that is unique.
    \item A subset of the image bits are uniquely fixed, and some remain arbitrary, resulting in a partial contraction map that may or may not admit a contracting completion.
    \item At least one image bit admits no solution to the contraction conditions, thus definitively implying that no contraction map exists.
\end{enumerate}

While doing so will not (and should not) yield a valid contraction map for arbitrary candidate inequalities, it will generate valid partial contraction mappings, unless a contradiction is reached. If a contradiction is reached, however, this contradiction is demanded by the consistency of the initial data, and therefore cannot be fixed with backtracking. Once all deterministic assignments have been made, then in principle the remaining choices are genuinely free, and backtracking should not be required prior to that point. It may nevertheless be that all choices after the deterministic fixing lead to contradictions, but empirically this has not occurred with the inequalities we studied, and nevertheless backtracking to prior to this point should not be required. 

\begin{definition}
    Consider a map $f: H_M \to H_N$ between hypercubes $H_M$ and $H_N$, and denote the distance function on $H_M$ and $H_N$ by $d_{\alpha}$ and $d_{\beta}$ respectively: 
    \begin{itemize}
        \item A pair of vertices $x,y\in H_N$ is said to be \textit{Hamming distance preserving} if $$d_{\alpha}(x,y)=d_{\beta}(f(x),f(y)).$$
        \item Given $x,y\in H_M$ such that $x_k=y_k$ if and only if $k\in K\subseteq[M]$, then another vertex $z\in H_M$ is said to be on a Hamming path between $x$ and $y$ if $$z_k=x_k=y_k \qquad \forall \ k\in K.$$
        In words: if $x$ and $y$ have the same bits in some sub-bitstring, then any $z$ sharing that same sub-bitstring is said to be on a Hamming path between $x$ and $y$.\footnote{Equivalently, $z$ is on a Hamming path between $x$ and $y$ if it lies on some path between them of path length equal to $d_{\alpha}(x,y)$. The set of all vertices on Hamming path between $x$ and $y$ defines the sub-hypercube where all $x_k$ are fixed for $k\in K$ and free for $k\in [M]\smallsetminus K$. Notice that, though possible, $z$ being on the Hamming path between $x$ and $y$ does not imply $y$ being on the Hamming path between $x$ and $z$.}
    \end{itemize}
\end{definition}
It is worth quoting at this point an important defining property distance functions obey. For any $x,y,z\in H_M$, the triangle inequality holds:
\begin{equation}
\label{eq:triangle}
    d_{\alpha}(x,y) \leq d_{\alpha}(x,z) + d_{\alpha}(y,z).
\end{equation}
Hereafter, all results we prove involving a map $f: H_M \to H_N$ implicitly assume this map is a contraction map, i.e., $f$ obeys $d_{\alpha}(x,y)\geq d_{\beta}(f(x),f(y))$ for all $x,y\in H_M$. In other words, we are proving general properties pertaining to the contraction maps relevant to proofs for holographic entropy inequalities using Theorem \ref{thm:contraction}.

\subsection{Deterministic Method 1}
From the definition above, a useful result follows:
\begin{lemma}
\label{lem:hdp}
    If $x,y\in H_M$ are Hamming distance preserving, then all vertices on Hamming paths between $x$ and $y$ map to vertices in $H_N$ on Hamming paths between $f(x)$ and $f(y)$.
\end{lemma}
\begin{proof}
    Let $x,y\in H_M$ be Hamming distance preserving, and consider another vertex $z\in H_M$. By letting $z$ be on a Hamming path between $x$ and $y$, we see that $d_{\alpha}(x,y)=d_{\alpha}(x,z)+d_{\alpha}(y,z)$.
    By assumption $d_{\alpha}(x,y)=d_{\beta}(f(x),f(y))$, and by the contraction condition it must be the case that $d_{\alpha}(x,z)+d_{\alpha}(y,z) \geq d_{\beta}(f(x),f(z))+d_{\beta}(f(y),f(z))$. Hence, in $H_N$ we obtain \begin{equation}
        d_{\beta}(f(x),f(y)) \geq d_{\beta}(f(x),f(z))+d_{\beta}(f(y),f(z)).
    \end{equation}
    However, as a distance function on $H_N$, $d_{\beta}$ obeys the triangle inequality \eqref{eq:triangle}, leading to
    \begin{equation}
        d_{\beta}(f(x),f(y)) = d_{\beta}(f(x),f(z))+d_{\beta}(f(y),f(z)).
    \end{equation}
    This in turn implies $f(z)$ is on a Hamming path between $f(x)$ and $f(y)$ in $H_N$, as claimed.
    
\end{proof}

The above result implies that any two Hamming distance preserving vertices $x,y\in H_M$ provide precise information about the image $f(z)$ for any vertex $z\in H_M$ on a Hamming path between them. Explicitly, for every $r\in[N]$ with coincident bit images $f(x)_r=f(y)_r\equiv b$, one can immediately assign the value $f(z)_r=b$. In fact, the requirement that $x$ and $y$ be Hamming distance preserving can be relaxed to obtain an even stronger result:

\begin{theorem}
\label{thm:genlemma}
    If $x,y\in H_M$ obey $0\leq d_{\alpha}(x,y)- d_{\beta}(f(x),f(y))\leq1$, then all vertices on Hamming paths between $x$ and $y$ map to vertices in $H_N$ on Hamming paths between $f(x)$ and $f(y)$.\footnote{Lemma \ref{lem:hdp} corresponds to the special case $d_{\alpha}(x,y)= d_{\beta}(f(x),f(y))$. This stronger result relaxes the Hamming distance preserving condition such that the conclusion holds more generally also when the distance on $H_M$ is reduced by at most one Hamming step on $H_N$.}
\end{theorem}

\begin{proof}
    Consider a vertex $z\in H_M$ on a Hamming path between $x$ and $y$, such that $d_{\alpha}(x,y)=d_{\alpha}(x,z)+d_{\alpha}(y,z)$. If $f(z)$ is on a Hamming path between $f(x)$ and $f(y)$, we would have $d_{\beta}(f(x),f(y)) = d_{\beta}(f(x),f(z))+d_{\beta}(f(y),f(z))$.
    Consider an alternative map $\tilde{f}$ yielding an image $\tilde{f}(z)$ off of Hamming paths between $f(x)$ and $f(y)$, but otherwise equal to $f$. This requires $\tilde{f}(z)$ having at least one bit flipped relative to $f(x)$ and $f(y)$ where the latter two coincide, and consequently also relative to $f(z)$. The resulting new map $\tilde{f}(z)$ will thus give distances obeying $d_{\beta}(\tilde{f}(x),\tilde{f}(z)) \geq d_{\beta}(f(x),f(z)) + 1$ and $d_{\beta}(\tilde{f}(y),\tilde{f}(z)) \geq d_{\beta}(f(y),f(z)) + 1$. Combining these, we obtain
    \begin{equation}
    \label{eq:nocont}
        \begin{aligned}
            d_{\beta}(\tilde{f}(x),\tilde{f}(z)) + d_{\beta}(\tilde{f}(y),\tilde{f}(z))
            & \geq d_{\beta}(f(x),f(z)) + d_{\beta}(f(y),f(z)) + 2 \\
            & \geq d_{\beta}(f(x),f(y)) + 2 \\
            & = d_{\alpha}(x,y) + 1 \\
            & = d_{\alpha}(x,z) + d_{\alpha}(y,z) + 1
        \end{aligned}    
    \end{equation}
    where the second line uses the triangle inequality on $H_N$, the third follows by hypothesis, and the fourth by the assumption that $z$ is on a Hamming path between $x$ and $y$. That the LHS of \eqref{eq:nocont} is strictly greater than $d_{\alpha}(x,z)+d_{\alpha}(y,z)$ is a direct violation of the contraction condition, implying that $\tilde{f}$ is not a contraction map.
\end{proof}

This result grants the following assignment rule to deterministically fix entries of a contraction map that uniquely solve the contraction conditions:

\begin{drule}
\label{rule1}
    For every $x,y\in H_M$ such that $0\leq d_{\alpha}(x,y)- d_{\beta}(f(x),f(y))\leq1$ and every $z \in H_M$ on a Hamming path between $x$ and $y$, Theorem \ref{thm:genlemma} uniquely fixes the following bits of $f(z)$:
    \begin{equation}
        f(z)_r = b \qquad \forall \ r\in[N] ~~ s.t. ~~ f(x)_r=f(y)_r\equiv b.
    \end{equation}
\end{drule}

This constraint can therefore be used to fix bits in bitstrings of $H_N$ on Hamming paths between $f(x)$ and $f(y)$ that satisfy the above condition. This provides a powerful method of directly assigning bits of the $H_N$ hypercube without the need to solve a naively NP-complete constraint satisfaction problem. Instead, one simply needs to determine all vertices along Hamming preserving paths of the initial data and fix their matched bits. Once new fully fixed bitstrings on the $H_N$ have been discovered, they may be appended to the initial data to search for new Hamming preserving paths.

The number of bits fixed by a single Hamming preserving path can be computed as follows: if $d_{\alpha}(x,y)=D$, and the dimensions of the $H_M$ and $H_N$ hypercubes are $M$ and $N$, respectively, then the number of fixed bits per row is $N-D$ and the number of rows for which fixing occurs is $2^D-2$. This means that the total number of fixed bits is $(2^D-2)(N-D)$. The total number of initially free bits is $(2^M-I) N$, where $I$ is the number of initial conditions. The ratio of these goes to one as $D$ approaches $M$ and if $D\ll N$. Note that this analysis does not address double counting of fixed data from subsequent Hamming preserving paths, so multiple paths would provide an overestimate.

\subsection{Deterministic Method 2}
A second constraint comes from considering when all degrees of freedom of a $H_N$ bitstring have been exhausted, in the sense that the contraction condition is saturated. In particular, consider a pair of vertices $x,y\in H_M$ such that their images $f(x)$ and $f(y)$ have been partially fixed (i.e., some image bits have been uniquely determined e.g. via Rule \ref{rule1}, but some remain undetermined). Then for any completion of the map $f$, we can lower-bound $d_{\beta}(f(x),f(y))$ by the partial distance $d^{\text{fixed}}_{\beta}(f(x),f(y))$ where $d_{\beta}$ is only applied to the dimensions $r\in[N]$ where both $f(x)_r$ and $f(y)_r$ have already been fixed. Explicitly, letting $I_x,I_y\subseteq[N]$ label the bits respectively in $f(x)$ and $f(y)$ that have already been uniquely fixed,
\begin{equation}
    d_{\alpha}(x,y)\geq d_{\beta}(f(x),f(y)) \geq d_{\beta}^{\text{fixed}}(f(x),f(y)) \equiv \sum_{r\in I_x\cap I_y} \beta_r |f(x)_r - f(y)_r|.
\end{equation}
This way one easily obtains the following rule:
\begin{drule}
\label{rule2}
    For every $x,y\in H_M$ such that $d_{\beta}^{\text{fixed}}(f(x),f(y)) = d_{\alpha}(x,y)$, the following bits get newly fixed:
    \begin{equation}
        f(x)_r=b \qquad \forall \ r\in I_y \smallsetminus I_x ~~ s.t. ~~ f(y)_r\equiv b, 
    \end{equation}
    and similarly under the exchange $x\leftrightarrow y$.
\end{drule}
In other words, the saturation of the contraction condition means that all remaining unfixed bits of $f(x)$ and $f(y)$ must match, which leads to uniquely determined images for all bits $f(x)_r$ and $f(y)_r$ with $r\in I_x \cup I_y$. While it is clear that unfixed bits on $[N] \smallsetminus (I_x \cup I_y)$ will necessarily also have to match between the two bitstrings, their specific value can however not be fixed by Rule \ref{rule2}.

This is a method that can be implemented after the first method to fix further bits. After it has been run, the first method can be run again as subsequent Hamming preserving paths are revealed. These methods are then alternated until no additional bits are fixed by either approach.

\subsection{Deterministic Method 3}

An additional definition is required to obtain our next rule:
\begin{definition}
    Given $x,y\in H_M$, another vertex $z\in H_M$ is said to be $k$-off Hamming paths between $x$ and $y$ if $k$ is the minimal distance from $z$ to any vertex $e$ on Hamming paths between $x$ and $y$,
    \begin{equation}
        k=\min_{e\in HP(x,y)}{d_{\alpha}(z,e)},
    \end{equation}
    where $HP(x,y)$ is the set of vertices appearing on the Hamming paths between x and y.
\end{definition}
This notion can be used to generalize Lemma \ref{lem:hdp} as follows:
\begin{lemma}
\label{lem:kohdp}
    If $x,y\in H_M$ are Hamming distance preserving, then any vertex that is $k$-off Hamming paths between $x$ and $y$ maps to a vertex in $H_N$ that is at most $k$-off Hamming paths between $f(x)$ and $f(y)$.
\end{lemma}
\begin{proof}
    Let $z\in H_M$ be a vertex $k$-off Hamming paths between $x,y\in H_M$. From the definition, let $e\in H_M$ be on a Hamming path between $x$ and $y$ such that $d_{\alpha}(z,e)=k$, and by the contraction condition it must be the case that $d_{\alpha}(z,e)\geq d_{\beta}(f(z),f(e))$, so we have $d_{\beta}(f(z),f(e))\leq k$. Since by Lemma \ref{lem:hdp} every $f(e)$ lies on a Hamming path between $f(x)$ and $f(y)$, we conclude that $f(z)$ is at most $k$-off Hamming paths between them.
\end{proof}

Now we can prove our next main result. Formally,

\begin{theorem}
    Let $z$ be a vertex on $H_M$ that is one-off from both a Hamming path between $x$ and $y$ and a Hamming path between $a$ and $b$ on $H_M$. Further, let $d_{\alpha}(x,y)=d_{\beta}(f(x),f(y))$ and $d_{\alpha}(a,b)=d_{\beta}(f(a),f(b))$. If $\exists$ a $H_N$ column $r$ such that $f(x)_r=f(y)_r\neq f(a)_r=f(b)_r$, then all matching bits between $f(a), f(b), f(x)$, and $f(y)$ are fixed for $f(z)$. Furthermore, $f(z)_r$ will not match one of $f(x)_r=f(y)_r$ or $f(a)_r=f(b)_r$, and whichever Hamming path it does not match in this column, $f(z)$ will match all remaining bits of that Hamming path.
\end{theorem}

\begin{proof}
    Let $z$ be a vertex on $H_M$ that is one-off from both a Hamming path between $x$ and $y$ and a Hamming path between $a$ and $b$ on $H_M$. Further, let $d_{\alpha}(x,y)=d_{\beta}(f(x),f(y))$ and $d_{\alpha}(a,b)=d_{\beta}(f(a),f(b))$. This forces by the previous lemma that $f(z)$ is at most one-off Hamming paths between $f(x)$ and $f(y)$, and also is at most one-off Hamming paths between $f(a)$ and $f(b)$. Furthermore, if $\exists$ a $H_N$ column $r$ such that $f(x)_r=f(y)_r\neq f(a)_r=f(b)_r$, then $f(z)$ is either exactly one-off Hamming paths between $f(x)$ and $f(y)$ or is exactly one-off Hamming paths between $f(a)$ and $f(b)$. This immediately requires that all matching bits between $f(a), f(b), f(x)$, and $f(y)$ are fixed for $f(z)$, as there are no remaining degrees of freedom for $f(z)$; if this were not the case, it would be at least two-off Hamming paths between $f(x)$ and $f(y)$ or at least two-off Hamming paths between $f(a)$ and $f(b)$. Furthermore, if WLOG $f(z)$ is exactly one-off $f(a)$ and $f(b)$, e.g. $f(z)_r=f(x)_r=f(y)_r\neq f(a)_r=f(b)_r$, then $f(z)$ must match all other matched bits between $f(a)$ and $f(b)$, for the same reason that all available degrees of freedom have been exhausted.
\end{proof}

The above generalizes easily for $k$-off for integer $k\geq 0$. Going to higher (for example six) point functions with three Hamming paths does not help, as any column can only be either zero or one, and so any higher point comparison would degenerate into comparisons of pairs of power sets of Hamming paths.

Indeed, the constraints generated by these three methods, and the third in particular, would seem to be complete, as any potential consequence of a contraction map constraint condition that yields an unconditionally fixed bit would fall into one of these three categories. Therefore, once these methods have been applied the remaining choices are observed to be free\footnote{We do not claim that these three rules are the complete set of rules. But empirically, they seem to fix (almost) all the deterministic bits and suffice to generate the contraction maps without backtracking.}. This has been empirically verified on all the five-party inequalities\cite{Hern_ndez_Cuenca_2019}, the known 384 six-party inequalities\cite{hubeny:2023hei} and the known seven-party inequalities\cite{Czech_2023}\cite{Czech:2023toric} for holographic entanglement entropy, as subsequent choices can be made without needing backtracking in all of these cases.

\subsection{Combining the Deterministic Methods}
Practically, the first and second methods both run in trivial amounts of time, while the third method is slower. The precise time complexity of these methods will be precisely described in Section \ref{sec:complexity}. Therefore, we adopt a strategy of alternating the first and second methods until stability has been reached before running the third method. After this, the first and the second methods are run until stability has been reached, at which point the third method is run again, with termination at a point where the third method does not fix any additional bits.

\subsubsection{Choice Constraints}
There is a constraint of a different type which we state as the following theorem:
\begin{theorem}\label{thm:choice_1}
   If $d_{\alpha}(x,y)=2$ or $d_{\alpha}(x,y)=3$ and $d_{\beta}(f(x),f(y))=0$, then if we take a vertex $z$ on the Hamming path between $x$ and $y$, then either $d_{\beta}(f(x),f(z))=0$ or $d_{\beta}(f(x),f(z))=1$.
\end{theorem}
\begin{proof}
    WLOG, consider $d_{\alpha}(y,z)=1$. Then, $d_{\alpha}(x,z)=1$ or $d_{\alpha}(x,z)=2$.
    By contraction map condition, $d_\beta(f(y),f(z))\leq 1$. Since $d_{\beta}(f(x),f(y))=0$, thus $d_\beta(f(x),f(z))\leq 1$.
\end{proof}
In other words, $z$ can only map to either the same vertex that both $x$ and $y$ map to or a neighbor of that vertex. Not doing so would lead to a violation of the contraction condition.
We introduce the following definition to state a related constraint.
\begin{definition}
    Given $x,y\in H_M$, $y$ is said to be distance-$k$-neighbor of $x$ (and vice-versa) if
    \begin{equation}
        d_{\alpha}(x,y)=k.
    \end{equation}
    For a given vertex $z$, we say that vertex $w \in N_{\alpha}(z,k)$, if $w$ is a distance-$k$-neighbor of $z$. We call the set $N_{\alpha}(z,k)$ as the set of distance-$k$-neighbors of $z$.
\end{definition}
The following constraint is useful to choose bitstrings that are neither fixed by the deterministic methods nor by the constraint in theorem \ref{thm:choice_1}.
\begin{theorem}\label{thm:choice_2}
  Consider a vertex $z \in H_M$. Given the set of distance-1-neighbors of $z$, $N_{\alpha}(z,1)\subset H_M$, define the set $I(z)\subset H_N$ such that
  \begin{equation}
      I(z)=\cap_{x\in N_{\alpha}(z,1)} \left( N_{\beta}(f(x),1) \cup f(x) \right)
  \end{equation} Then, $f(z)\in I(z)$.
\end{theorem}
\begin{proof}
    Given the vertex $z$, consider two distinct vertices $x,y\in N_{\alpha}(v,1)$, i.e, $d_{\alpha}(x,z)=1$ and $d_{\alpha}(y,z)=1$. By contraction map condition, $d_{\beta}(f(x),f(z))\leq1$ and $d_{\beta}(f(y),f(z))\leq1$, i.e, $f(z)\in N_{\beta}(f(x),1) \cup f(x)$ and $f(z)\in N_{\beta}(f(y),1) \cup f(y)$. This gives us $f(z)\in \left(N_{\beta}(f(x),1) \cup f(x)\right) \cap \left(N_{\beta}(f(y),1) \cup f(y)\right)$. Since $x, y$ are arbitrary, by induction $f(z)\in I(z)$.
\end{proof}
Theorem \ref{thm:choice_2} can be generalized for up to distance-$k$-neighbors. Modulo these conditional constraints, one then simply freely chooses an unfixed bit (say the top-left bit in the tabular representation of $H_N$ and fixes it to be $1$ or $0$). For technical reasons, fixing it to be $1$ results in faster convergence than fixing it to be $0$, as the $H_N$ deterministically fixed bitstrings tend to have more $0$'s than $1$'s. After this bit is fixed, the deterministic fixer methods are run again until stability is achieved, at which point another choice is made. Choices are made in this way until a contradiction is reached, or a full contraction map has been specified. It is interesting to quantify the number of choices that must be made in any given known valid inequality.

\subsubsection{Examples}
As an example, consider the five-party inequality given by 

\begin{equation}
\begin{split}
3S(ABC) + 3S(ABD) + 3S(ACE) + S(ABE) + S(ACD) + S(ADE) + S(BCD) + \\
S(BCE) + S(BDE) + S(CDE) \geq 2S(AB) + 2S(ABCD) + 2S(ABCE) + 2S(AC) + \\
2S(BD) + 2S(CE) + S(ABDE) + S(ACDE) + S(AD) + S(AE) + S(BC) + S(DE).\\    
\end{split}\label{5.5}
\end{equation}
This inequality has $10$ terms on the LHS and $18$ terms on the expanded RHS. The contraction between hypercubes is a map $f:2^{10}\rightarrow 2^{18}$. There are $1024$ bitstrings, each having $18$ bits on the RHS totalling to $18432$ bits. Before making any free choice, the deterministic methods $1$ and $2$ fully fixes $170$ bitstrings and partially fixes $842$ more bitstrings, leaving a total of $3636$ bits unfixed. Applying method $3$, the number of unfixed bits is reduced to $3558$. The greedy algorithm leaves $3588$ bits unfixed, i.e, it fixes $30$ less bits than our methods $1$, $2$ and $3$ combined.\footnote{Note that the improvement in performance of our algorithm compared to the greedy algorithm is specific to this example. For most cases, they fix the same number of bits, except in some cases where it fixes slightly more and in a few cases, it fixes slightly less, hinting that the rules are incomplete.} This inequality has a \emph{maximal choice}, or upper bound of choices that must be made if the choices are made to require the largest number of subsequent choices, of 142 bits choosing them all to be 0 and choosing an entire string of 18 bits. The true number of choices are less than the maximal choice (for example, choosing 1 at every choice instead of 0 ends up with a total of 83 bit choices instead). This number of choices is representative of the number that must be made for the other inequalities at hand. We summarize the details of the contraction maps for the five five-party inequalities found in \cite{Bao_2015} in the table \ref{tab:5party_stats}.\footnote{We attach a few more examples of known six-party and seven-party inequalities in the Appendix \ref{app:a}.}
The four other five-party inequalities (\ref{5.5} being the fifth one) are as follows
\begin{equation}
 \begin{split}
& S(ABC) + S(ABD) + S(ACE) + S(BCD) + S(BCE) \geq \\ & S(A) + S(BC) + S(BD) + S(CE) + S(ABCD) + S(ABCE)\label{5.1}
\end{split} \end{equation}
\begin{equation}
 \begin{split}
& S(AD) + S(BC) + S(ABE) + S(ACE) + S(ADE) + S(BDE) + S(CDE) \geq \\ & S(A) + S(B) + S(C) + S(D) + S(AE) + S(DE) + S(BCE) + S(ABDE) + S(ACDE)\label{5.2}
\end{split} \end{equation}
\begin{equation}
 \begin{split}
& 2S(ABC) + S(ABD) + S(ABE) + S(ACD) + S(ADE) + S(BCE) + S(BDE) \geq \\ & S(AB) + S(AC) + S(AD) + S(BC) + S(BE) + S(DE) + S(ABCD) + S(ABCE) + S(ABDE)\label{5.3}
\end{split} \end{equation}
\begin{equation}
 \begin{split}
& S(ABC) + S(ABD) + S(ABE) + S(ACD) + S(ACE) + S(ADE) + S(BCE) + S(BDE) \\ & + S(CDE) \geq  S(AB) + S(AC) + S(AD) + S(BE) + S(CE) + S(DE) + S(BCD)  \\ &  + S(ABCE) + S(ABDE) + S(ACDE) \label{5.4}
\end{split} \end{equation}

\begin{table}[h!]
\centering
\begin{tabular}{@{}llllll@{}}
\toprule
\textbf{Inequality} &
  \textbf{M} &
  \textbf{N} &
  \textbf{\begin{tabular}[c]{@{}l@{}}\% bits fixed by\\ methods 1 \& 2\\ until first choice\end{tabular}} &
  \textbf{\begin{tabular}[c]{@{}l@{}}No. of \\ bit\\ choices\end{tabular}} &
  \textbf{\begin{tabular}[c]{@{}l@{}}No. of \\ string\\ choices\end{tabular}} \\ \midrule
\ref{5.1} & 5  & 6  & 94.79 & 5   & 0  \\
\ref{5.2} & 7  & 9  & 77.60 & 35  & 0  \\
\ref{5.3} & 7  & 9  & 93.75 & 18  & 0  \\
\ref{5.4} & 9  & 10 & 93.78 & 22  & 12 \\
\ref{5.5} & 10 & 18 & 80.27 & 142 & 1  \\ \bottomrule
\end{tabular}
\caption{Summary of 5-party contraction maps with LHS and RHS of length M and N respectively. This table shows the percentage of bits fixed by deterministic run from methods 1 and 2 before making the first choices and finally notes the number of choices made to generate a contraction map.}
\label{tab:5party_stats}
\end{table}

We also give two examples of previously unknown six-party facet inequalities\footnote{Our proof method has proved 1116 unpublished 6-party inequality candidates shared by Sergio Hernández-Cuenca in private communication, which are computationally intractable by the state-of-the-art greedy algorithm. We thank Sergio Hernández-Cuenca for sharing and allowing to use them.} that we proved using contraction map method, in \ref{ineq:new_6.1} and \ref{ineq:new_6.1052},
\begin{equation}
 \begin{split}
& 2S(ABC) + 2S(ABD) + S(ABE) + S(ABF) + S(ACD) + 2S(ACE) + S(ACF) + S(ADE) \\ & + S(ADF) + 2S(BCD) + S(CDE) + S(CEF) \geq 2S(AB) + S(AC) + S(AD) + S(AE) \\ &+ S(AF) + S(BC) + S(BD) + S(CD) + S(CE) + S(CF) + S(DE) + 2S(ABCD) \\ &+ S(ABCE) + S(ABDF) + S(ACDE) + S(ACEF).\label{ineq:new_6.1}
\end{split} \end{equation}

\begin{equation}
 \begin{split}
& S(ABC) + S(ABD) + S(ABE) + S(ABF) + S(ACD) + S(ACF) + S(ADE) + S(AEF) \\ &+ S(BCF) + S(BDF) + S(CDF) + S(ABCE) + S(BCEF) + S(CDEF) \geq \\ & 2S(AB) + S(AC) + S(AD) + S(AE) + S(BF) + S(CD) + S(CF) + S(DF) + S(EF) \\ &+ S(BCE) + S(ABCF) + S(ABDE) + S(BCDF) + S(ABCEF) + S(ACDEF)\label{ineq:new_6.1052}
\end{split} \end{equation}

\subsection{Aside: Unphysical $H_M$ Vertices}
As was first pointed out in \cite{Avis_2023}, certain vertices in $H_M$ are unphysical; for example, a vertex cannot be included in $A$ but excluded in $AB$. By nesting of minimal-cardinality min-cuts, such vertices simply do not exist (holographically, entanglement wedge nesting guarantees that such spacetime regions are empty). These vertices can formally be characterized in terms of intersecting anti-chains among the collections of subsystems they label (we refer the interested reader to \cite{Avis_2023} for more details on this classification).\footnote{A more holography-based approach to these unphysical vertices was pursued in \cite{Li_2022}, which we believe to be complementary to our main techniques.}

It is therefore a question as to whether these vertices must be mapped to $H_N$ and whether if they were mapped to $H_N$, they would result in over-stringent (and undesired) constraints for the contraction proof method. By analysis of the known inequalities, the answer to both of these questions appears to be no. Regarding the former, at the level of the proof methods, it is clearly unnecessary to include such vertices in $H_N$: the graph vertex they label does not exist, so no edges attached to it exist either, and thus there is no need to rearrange their contribution to left-hand-side cuts into contributions to right-hand-side cuts. In other words, a contraction map for $H_M$ without unphysical bitstrings included suffices to prove the corresponding inequality valid. In principle, it could be a logical possibility for a contraction map to exist for $H_M$ without unphysical bitstrings, but to not exist for the full $H_M$. This is the latter question, which we experimentally answer in the negative. For all known valid inequalities, contraction maps exist regardless of whether unphysical bitstrings are included. Something even stronger happens to hold: the inclusion of unphysical bitstrings turns out to not enforce the deterministic fixing of any single additional bit on the right-hand side. This suggests that the contraction conditions following from unphysical bitstrings are always strictly weaker than the rest, and thus redundant. We leave the proof of this statement for future work. Whether or not true, it is clear that at both fundamental and practical levels, removing unphysical vertices from $H_N$ is preferred.

It is an interesting question to ask if the number of such unphysical bitstrings decreases as the party number increases. While naively we have a way of analyzing this via the infinite family of cyclic holographic entanglement entropy inequalities, as an independent member of this family appears for each odd party number, the LHS terms in the cyclic family of interest all together form an intersecting anti-chain (every term has more than 1/2 of the terms and they all cross), which means any subset of them is also an intersecting anti-chain. In other words, there are no unphysical bitstrings for the cyclic family. However, we find that statistically, the number of unphysical bitstrings is significantly higher for the six-party inequalities than for the five-party ones, and establishing whether this trend continues to higher party numbers is a useful course of study.

We can make this more precise with the following estimate. The unphysicality of bitstrings on $H_M$ can be studied pairwise via the columns of the contraction map, as for every pair of columns either one is contained in the other, or the intersection is trivial. As these possibilities are the only ones that would lead to a bitstring being deemed unphysical, and as for each of these only one combination of two bits yields an unphysical outcome ($11$ for trivial intersection and $10$ for containment, WLOG), when an offending combination occurs a quarter of the naive bitstrings are rendered unphysical.

Consider a pair of columns defining $H_M$. Let one of them be an $a$-party column, and the other be an $b$-party column. We can take $a$ and $b$ to both be less than or equal to $N$, the party number of the inequality, by purification. There are $N \choose a$ ways different possibilities for the $n$-party column; WLOG we can specialize to a specific one. Then, ${N-a} \choose b$ out of $N \choose b$ choices for the $m$ party column have a trivial intersection with the $a$ party column. For $a, b\ll N$,\footnote{We note that, by symmetry, if $N-a,N-b\ll N$ the same result holds.} this is very close to all of them. Each of these results in only $3/4$ of the bitstrings having the possibility of remaining physical. This is, however, only the analysis for trivial intersection,\footnote{By symmetry the analysis for containment yields the same result.} and only for a single pair of columns. Once all columns with small $a,b$ are considered, the number of such pairs selected without replacement\footnote{The without replacement is for reasons of independence; the fraction changes a bit away from $3/4$ per pair if for example, the same $a$-party column has a trivial intersection with multiple $b$-party columns. In this case, where a single $a$-party column has a trivial intersection with $k-1$ $b$-party columns, the combined fraction of physical bitstrings retained is $1-\frac{1}{2^k}$.} gives the exponent that the $3/4$ is raised to, and will generically eliminate almost all naive bitstrings as unphysical.

Even if only one of $a,b\ll N$, and the other is as large as is allowed by symmetry to make the reduction as small as possible, e.g. $a=N/2$, if $b\ll N$ then the fraction of choices of the $b$-party column that will result in some unphysicality as above goes as $\frac{1}{2^b}$, which for small $b$ is nontrivial. It's only when $a,b\gg1$ and at least one of $a,b\sim N/2$ that the fraction of choices resulting in trivial intersections approaches zero.

\section{Complexity and Completeness}
\label{sec:complexity}

\subsection{Complexity}
Recall $H_M$ and $H_N$ are Hamming hypercubes on left and right-hand sides of an inequality, where we have $|H_M|=2^M=\mathcal{N}$. Both $M$ and $N$ are of the order of $\log{\mathcal{N}}$. If one were to assign all the RHS vertices with the same entry, one would have to fill up $M$ values $2^N$ times, suggesting that the most trivial toy map (not necessarily a contraction) will have a complexity $\mathcal{O}(\mathcal{N}\log{\mathcal{N}})$. Any realistic contraction map has a computational complexity greater than this.

The computational complexity of applying methods 1 and 2 are upper bounded by $\mathcal{O}(\mathcal{N}^3\log{\mathcal{N}})$ and $\mathcal{O}(\mathcal{N}^2\log{\mathcal{N}})$ respectively. One executes them sequentially starting from initial data to fixing a certain number of bits, where one has to make a choice to fix any more bits. Note that calculation of this upper bound assumes the distance between two bitstrings in LHS is maximal $(M)$ and every pair of bitstrings are Hamming-distance preserving from $H_M$ to $H_N$, both of which are over-estimations. We call the complexity associated with making choices a query complexity $\mathcal{Q}$.
It is empirically observed that $\mathcal{Q}\sim\mathcal{O(\mathcal{N})}$.

For most inequalities, methods 1 and 2 suffice to fix as many bits as including method 3 would, i.e, method 3 doesn't fix any additional bits (for some inequalities, method 3 fix additional bits). As the methods 1 and 2 are executed after making every bit choice, the estimated upper bound of complexity to generate the entire contraction map is $\mathcal{O}(\mathcal{N}^4\log{\mathcal{N}})$.

To stubbornly apply method 3, one rarely fixes a large number of bits while increasing the computational complexity greatly. One first has to find all pairs of Hamming distance preserving vertices fixed by methods 1 and 2, whose complexity scales as $\mathcal{O}(\mathcal{N}^2 \log{\mathcal{N}})$. By our algorithmic implementation of a single run of method 3 using two unique Hamming distance preserving pairs of vertices $(x_1, x_2)$ and $(x_3, x_4)$, the complexity is bounded above by $\mathcal{O}(\mathcal{N}^4 (\log{\mathcal{N}})^2)$. Running through all such pairs, the complexity scales as $\mathcal{O}(\mathcal{N}^8 (\log{\mathcal{N}})^2)$. Note that this again an inflated upper bound as we are assuming that the number of Hamming distance preserving pairs scale as $\mathcal{O}(\mathcal{N}^2)$ (while empirically it has been observed to be of the order of $\mathcal{O}(\mathcal{N})$). We apply method 3 only once after the first deterministic run and rely on the power of methods 1 and 2 once the algorithm starts making choices.\footnote{In the more recent implementations, we have relied entirely on methods 1 and 2. As long as the \textit{few} deterministic \textit{unfixed} bits are fixed by some Hamming preserving path resulting from prior free choices, it is not an issue.}

The above routines run until it fixes no new bits. The number of choices that must be made can be understood as the number of queries needed to complete a contraction proof. While there is no theoretical bound (short of the total number of unfixed bits initially) for the number of queries, in practice the number of queries required is significantly lower, and does not exceed $\mathcal{O(N)}$ choices for any given inequality to date. This number is also not optimized, and so in principle could be even lower, given a more optimized 0/1 selection strategy when choices are made.\footnote{We have implemented a version of our code using only the physical bitstrings and it reduces the run-time by a huge factor whenever the number of physical bitstrings $\mathcal{N}_{phy}$ is less than the total number of vertices of the hypercube $\mathcal{N}$ by a significant factor.} A pseudocode summarizing the working of the algorithm is given in Algorithm 1.

\subsection{Completeness}

As the constraints of the contraction map are built out, sub-graphs of the final graph are constructed step by step, with unfixed bits corresponding to regions for which not all adjacency conditions via wormhole are known. Such unknowns correspond to degrees of freedom remaining in the contraction map.

This leads to a path towards a proof of completeness of the contraction map method. Because every fixing of the bitstring reduces the remaining degrees of freedom of the graph, and thus of a wormhole geometry, one can ask what happens when the contraction map fails. If the contraction map is implemented bit by bit, this must correspond to a situation where a single bit is forced to be both $1$ and $0$, which it cannot be. What this would mean is that two different graph extensions of the graph corresponding to the bit just prior to the contradictory bit are generated. These graphs will serve as the counterexample graphs, where their cuts will fail the contraction map, which corresponds to the cutting and pasting strategy for all holographic entanglement entropy/graph cut inequalities. Something that is still lacking is a constructive algorithm for generating the graph that serves as the counterexample to a particular false inequality. While our arguments here show that such a graph must exist, it does not explicitly generate such a graph from knowledge of the inequality alone. We leave the development of this constructive algorithm for future work.

Also, there is an argument that the methods we used and straightforward generalizations thereof encompass all possible deterministic fixings, specifically via combinations of Hamming-preserving paths. The argument proceeds by contradiction, by assuming that there exists a rigid sub-mapping between hypercubes that is not a Hamming-preserving path (or an off-by-one by previous parity arguments). However, such a mapping could always be reduced to a Hamming-preserving path by modification of the LHS to take up the slack. Therefore, such a sub-mapping would not be rigid, as a rigid mapping would not allow for such slack by definition, as such slack would always be associated with a non-forced choice. Therefore, the only rigid maps permitted are Hamming distance preserving maps, and therefore our deterministic mapping methods, appropriately generalized to $n$ Hamming distance preserving maps, is complete. We leave the formal proof of this statement for future work.

\section{Discussion}
\label{sec:discussion}

It is a tantalizing direction to consider what these deterministic methods would mean for the bulk spacetime in AdS/CFT directly. The Hamming-preserving maps correspond to isometries between certain sub-graphs of $H_M$ to certain sub-graphs of $H_N$. Because these hypercubes are representations of the spacetime itself, it suggests isometries between sequences of certain bulk regions separated by Ryu-Takayanagi surfaces associated with portions of the LHS and RHS of the entropy inequalities. This has the potential to give novel constraints regarding the metric rigidity of AdS/CFT, possibly connecting to bulk metric reconstruction \cite{Bao_2019}.

A major limitation of this work is that it gives no aid in generating candidate inequalities for holographic entanglement entropy. Therefore, a complementary future direction to this work is to find a way of efficiently generating candidate inequalities or to generate automatically correct inequalities that are true, but whose tightness to the cone must be checked. To do this, one can potentially study the problem of finding all contraction maps, or at least some families of maps that scale with $N$ and $M$ between given hypercubes $M$ and $N$ with no single-party constraints. Once such families have been found, the single-party bitstrings of $H_M$ and $H_N$ can be assigned retroactively, which then specifies the entropies that appear in the candidate inequality. Finding such families of contraction maps appears to be a difficult combinatorial problem, but a clever solution thereof would immediately generate novel families of entanglement entropy inequalities.

\acknowledgments
We thank Sergio Hernandez-Cuenca for initial collaboration. We thank Scott Aaronson, Bartek Czech, Miles Cranmer and Jason Pollack for discussions during the writing of this paper. We would also like to thank Xi Dong, Gabriel Trevi\~no, Michael Walter, and Wayne Weng for earlier discussions on similar ideas. N.B. is supported by the Computational Science Initiative at Brookhaven National Laboratory, Northeastern University, and by the U.S. Department of Energy QuantISED Quantum Telescope award. J.N. is supported by the graduate assistantship at Northeastern University.

\begin{algorithm}
    \caption{Contraction map pseudocode}
    \label{alg:pseudocode}
    \begin{algorithmic}[1]
        \Procedure{Initialization}{}
            \State Read the inequality, store the terms and coefficients.
            \State Generate a hypercube $H_M$ canonically, assign LHS Hamming weights.
            \State Expand the RHS coefficients to $N$ columns and generate blank images in $H_N$.
            \State Incorporate the initial conditions in the candidate map.
        \EndProcedure
        \State
        \Procedure{Deterministic Fixing}{}
        \While{No more deterministic fixing possible,}
            \State Run Method 1.
            \State Run Method 2.
        \EndWhile
        \State Run Method 3. (Optional).
        \If{All bits fixed}
                    \State Check contraction map consistency.
                    \If{Success}
                        \State \textbf{Proof completed. Contraction map found.}
                    \Else
                    \State \textbf{Proof failed. Contraction map doesn't exist.}
                    \EndIf
        \Else
            \State Check partial contraction map consistency.
            \If{Failure}
            \State \textbf{Proof failed. Contraction map doesn't exist.}
        \EndIf
        \EndIf
        \EndProcedure
        \State
        
        \Procedure{Bit-wise Choices}{}
        \State Set iterator $k$ value at 0 and choose maximum number of iterations $k_{max}$.
        \While{$\exists$ unfixed bit \textbf{and} $k<k_{max}$}
            \State Iterator $k+=1$
            \State Locate the top-left-most unfixed bit $j$ in a partially fixed string $i$.
            \While{Bitstring $i$ is unfixed}
                \State Choose $0$.
                \State Run Method 2 on string $i$ alone.
            \EndWhile
            \State \textbf{Run Procedure} \texttt{Deterministic fixing} (check consistency only in the end)
        \EndWhile
        \EndProcedure
    \State
    
    \Procedure{String and Bit Choices}{}
    \For{$i = 0;\ i < 2^M-1;\ i++$}
    \If{Bitstring $i$ is fully unfixed}
            \State \textbf{Choose} string.
        \Else
            \State \textbf{Run Procedure} \texttt{Bit-wise Choices} to string $i$ alone.
        \EndIf
    \EndFor
    \EndProcedure
    \end{algorithmic}
\end{algorithm}

\newpage
\appendices
\section{More examples of known holographic entropy inequalities}\label{app:a}
\subsection{Six-party inequalities}
The following six-party inequalities are adapted from \cite{hubeny:2023hei} (see table \ref{tab:6party_stats} for a summary of their contractions maps),

\begin{equation}
\begin{split}
& S(ABC) + S(ABD) + S(ABE) + S(ACD) + S(ACF) + S(BCEF) \geq \\ & S(AB) + S(AC) + S(AD) + S(BE) + S(CF) + S(ABCD) + S(ABCEF)\label{6.1}
\end{split} \end{equation}
\begin{equation}
 \begin{split}
& S(BC) + S(ABD) + S(ABE) + S(ACD) + S(ADF) + S(BDF) + S(BCDE) \geq \\ & S(A) + S(B) + S(C) + S(AD) + S(BE) + S(DF) + S(BCD) + S(ABDF) + S(ABCDE)\label{6.2}
\end{split} \end{equation}
\begin{equation}
 \begin{split}
& 2S(ABC) + S(ABD) + S(ABE) + S(ACD) + S(ACF) + S(BCE) + S(BCF) \geq \\ & S(AB) + S(AC) + S(AD) + S(BC) + S(BE) + S(CF) + S(ABCD) + S(ABCE) + S(ABCF)\label{6.3}
\end{split} \end{equation}
\begin{equation}
 \begin{split}
& S(ABC) + S(ABD) + S(ABE) + S(ACD) + S(ACF) + 2S(BCD) + S(BCEF) \geq \\ & 2S(A) + S(BC) + S(BD) + S(BE) + S(CD) + S(CF) + 2S(ABCD) + S(ABCEF)\label{6.4}
\end{split} \end{equation}
\begin{equation}
 \begin{split}
& S(ABC) + S(ACD) + S(ACE) + S(ADF) + S(BDE) + 2S(CDE) + S(BCDF) \geq S(A)  \\ & + S(B) + S(AC) + S(CD) + S(CE) + S(DE) + S(DF) + S(ACDE) + S(BCDE) + S(ABCDF)\label{6.5}
\end{split} \end{equation}
\begin{equation}
 \begin{split}
& S(ABC) + S(ABD) + S(ABE) + S(ABF) + S(ACD) + S(ACE) + S(BCE) + S(BEF) \geq \\ & S(AB) + S(AC) + S(AD) + S(BE) + S(BF) + S(CE) + S(ABCD) + S(ABCE) + S(ABEF)\label{6.6}
\end{split} \end{equation}
\begin{equation}
 \begin{split}
& S(BC) + S(ABD) + S(ACD) + S(AEF) + S(BCD) + S(BDE) + S(CDE) + S(DEF) \geq \\ & S(A) + S(B) + S(C) + S(BD) + S(CD) + S(DE) + S(EF) + S(ABCD) + S(ADEF) + S(BCDE)\label{6.7}
\end{split} \end{equation}
\begin{equation}
 \begin{split}
& S(ABC) + S(ABD) + S(ACE) + S(ACF) + S(BCE) + S(ABDE) + S(ACDE) + S(BCDF) \geq \\ & S(A) + S(B) + S(AC) + S(BD) + S(CE) + S(CF) + S(ADE) + S(ABCE) \\ & + S(ABCDE) + S(ABCDF)\label{6.8}
\end{split} \end{equation}
\begin{equation}
 \begin{split}
& S(ABC) + S(ABD) + S(ACE) + S(AEF) + S(BCE) + S(BDE) + S(DEF) + S(ACDE) \geq \\ & S(A) + S(B) + S(AC) + S(BD) + S(CE) + S(DE) + S(EF) + S(ABCE) + S(ADEF) + S(ABCDE)\label{6.9}
\end{split} \end{equation}
\begin{equation}
 \begin{split}
& S(ABC) + S(ACD) + S(ACE) + S(AEF) + S(BDF) + S(CDE) + S(CDF) + S(BCEF) \geq \\ & S(A) + S(B) + S(AC) + S(CD) + S(CE) + S(DF) + S(EF) + S(ACDE) + S(BCDF) + S(ABCEF)\label{6.10}
\end{split} \end{equation}

\begin{longtable}{@{}llllll@{}}
\toprule
\textbf{Inequality} & \textbf{M} & \textbf{N} & \textbf{\begin{tabular}[c]{@{}l@{}}\% bits fixed by\\ methods 1 \& 2\\ until first choice\end{tabular}} & \textbf{\begin{tabular}[c]{@{}l@{}}No. of\\ bit\\ choices\end{tabular}} & \textbf{\begin{tabular}[c]{@{}l@{}}No. of\\ string\\ choices\end{tabular}} \\ \midrule
\ref{6.1} & 6 & 7 & 91.74 & 14 & 0 \\
\ref{6.2} & 7 & 9 & 81.94 & 35 & 2 \\
\ref{6.3} & 7 & 9 & 93.75 & 18 & 0 \\
\ref{6.4} & 7 & 10 & 86.01 & 37 & 1 \\
\ref{6.5} & 7 & 10 & 91.79 & 22 & 1 \\
\ref{6.6} & 8 & 9 & 86.54 & 32 & 0 \\
\ref{6.7} & 8 & 10 & 72.42 & 53 & 4 \\
\ref{6.8} & 8 & 10 & 78.78 & 77 & 5 \\
\ref{6.9} & 8 & 10 & 79.64 & 55 & 1 \\
\ref{6.10} & 8 & 10 & 85.0 & 37 & 0 \\ \bottomrule
\caption{Summary of 6-party contraction maps with LHS and RHS of length M and N respectively. This table shows the percentage of bits fixed by deterministic run from methods 1 and 2 before making the first choices and finally notes the number of choices made to generate a contraction map.}
\label{tab:6party_stats}
\end{longtable}

\subsection{Seven-party inequalities}
We will use the two seven-party inequalities from \cite{Bao_2015} and \cite{Czech_2023}, whose contraction maps we summarize in table \ref{tab:7party_stats},
\begin{equation}
\begin{aligned}
  & S(ABCD)+ S(ABCG)+ S(ABFG)+ S(AEGF)+ S(BCDE)+ S(CDEF)+ S(DEFG) \geq
  \\ & S(ABC)+ S(ABG)+ S(AFG)+ S(BCD)+ S(CDE)+ S(DEF)+ S(EFG)+ S(ABCDEFG),
  \label{7.1}
\end{aligned}
\end{equation}

\begin{equation}
\begin{aligned}
  & S(ABDE)+S(ABDF)+S(ABEG)+S(ACDE)+S(ACDF)+S(ACEG)+S(ADEF)+S(ADEG) \\
  & +S(BCDE)+S(BCDF)+S(BCEG)+S(BDEF)+S(BDEG)+S(CDEF)+S(CDEG) \geq \\
  & S(ABC)+S(ADE)+S(ADF)+S(AEG)+S(BDE)+S(BDF)+S(BEG)+S(CDE)+S(CDF) \\
  & +S(CEG)+S(ABDEF)+S(ABDEG)+S(ACDEF)+S(ACDEG)+S(BCDEF)+S(BCDEG).
  \label{7.2}
\end{aligned}
\end{equation}

\begin{table}[h!]
\centering
\begin{tabular}{@{}llllll@{}}
\toprule
\textbf{Inequality} & \textbf{M} & \textbf{N} & \textbf{\begin{tabular}[c]{@{}l@{}}\% bits fixed by\\ methods 1 \& 2\\ until first choice\end{tabular}} & \textbf{\begin{tabular}[c]{@{}l@{}}No. of\\ bit\\ choices\end{tabular}} & \textbf{\begin{tabular}[c]{@{}l@{}}No. of\\ string\\ choices\end{tabular}} \\ \midrule
\ref{7.1} & 7 & 8 & 92.48 & 10 & 0 \\
\ref{7.2} & 15 & 16 & 89.41 & 705 & 108 \\ \bottomrule
\end{tabular}
\caption{Summary of 7-party contraction maps in the format of table \ref{tab:6party_stats}.}
\label{tab:7party_stats}
\end{table}


\bibliographystyle{JHEP}
\bibliography{main.bib}

\end{document}